\documentclass[10pt,tightenlines,showpacs,letterpaper,aps,prd,longbibliography,twocolumn,nofootinbib,superscriptaddress]{revtex4-2}
\usepackage{rotating}
\usepackage{floatpag}
\rotfloatpagestyle{empty}

\usepackage{dsfont}
\usepackage{amsmath,amsthm,amssymb}
\usepackage{xspace,enumerate,epsfig}
\usepackage{graphicx}
\graphicspath{{.}{./figures/}}
\usepackage{marvosym}

\let\olddagger\dagger
\renewcommand{\dagger}{\ensuremath{\olddagger}\xspace}

\theoremstyle{definition}
\newtheorem{theorem}{Theorem}
\newtheorem{corollary}[theorem]{Corollary}
\newtheorem{lemma}[]{Lemma}

\newtheorem{definition}[]{Definition}

\newtheorem{example*}[theorem]{Example*}
\newtheorem{examples*}[theorem]{Examples*}

\newtheorem{remark*}[theorem]{Remark*}

\theoremstyle{plain}

\usepackage{color}
\def\bR{\begin{color}{red}}
\def\bB{\begin{color}{blue}}
\def\bM{\begin{color}{magenta}}
\def\bC{\begin{color}{cyan}}
\def\bW{\begin{color}{white}}
\def\bBl{\begin{color}{black}}
\def\bG{\begin{color}{green}}
\def\bY{\begin{color}{yellow}}
\def\e{\end{color}\xspace}
\newcommand{\bit}{\begin{itemize}}
\newcommand{\eit}{\end{itemize}\par\noindent}
\newcommand{\ben}{\begin{enumerate}}
\newcommand{\een}{\end{enumerate}\par\noindent}
\newcommand{\beq}{\begin{equation}}
\newcommand{\eeq}{\end{equation}\par\noindent}
\newcommand{\beqa}{\begin{eqnarray*}}
\newcommand{\eeqa}{\end{eqnarray*}\par\noindent}
\newcommand{\beqn}{\begin{eqnarray}}
\newcommand{\eeqn}{\end{eqnarray}\par\noindent}

\def\jR{\begin{color}{DarkRed}}
\def\jB{\begin{color}{MidnightBlue}}
\def\jM{\begin{color}{magenta}}
\def\jC{\begin{color}{cyan}}
\def\jW{\begin{color}{white}}
\def\jBl{\begin{color}{black}}
\def\jG{\begin{color}{green}}
\def\jY{\begin{color}{yellow}}

\def\cR{\begin{color}{Crimson}}
\def\cB{\begin{color}{cyan}}
\def\cM{\begin{color}{magenta}}
\def\cC{\begin{color}{cyan}}
\def\cW{\begin{color}{white}}
\def\cBl{\begin{color}{black}}
\def\cG{\begin{color}{green}}
\def\cY{\begin{color}{yellow}}

\usepackage{amssymb}
\usepackage{epstopdf}
\usepackage{amsmath,amsthm}	
\usepackage{bm}
\usepackage{etex}
\usepackage{hyperref}

\newcounter{procedurecount}
\newcommand{\protocoltitle}{}

\begin{document}

        \title{Operational Discriminability: From Noncontextuality Bounds to Bell Correlations}

\author{Seyed Arash Ghoreishi\thanks{arash.ghoreishi@savba.sk}}
 
 \affiliation{Institute of Physics, Slovak Academy of Sciences, Dúbravská Cesta 9, 84511 Bratislava, Slovakia}

\begin{abstract}
We investigate discriminability from an operational and contextuality-oriented perspective using a two-copy comparison game based on SWAP-type measurements. The resulting score $D_{\mathrm{op}}$ provides an experimentally accessible notion of distinguishability that does not rely on a minimum-error discrimination task. We first examine whether this discriminability game can directly witness preparation contextuality. Within a preparation-noncontextual ontological model, we derive a direct upper bound on the game score under a SWAP-like comparison rule and a sharp single-copy test, and show that this bound is saturated in the natural qubit realization. Thus, the direct game alone does not provide a contextuality witness in that regime. We then consider a Bell-coupled scenario in which two-copy comparison measurements are applied to Bob's conditional preparations. This yields a state-dependent upper bound on the CHSH value in terms of operational separation parameters, and hence in terms of the distinguishability of the conditional states. Our results establish a quantitative link between operational discriminability
and the strength of nonclassical correlations, showing that discriminability can act as an operational resource constraining Bell-type nonclassical
correlations.
\end{abstract}

\maketitle
\section{Introduction}

The ability to distinguish physical states is a central concept in quantum information theory. 
A wide range of protocols, including quantum communication, cryptography, and metrology, rely on the fact that nonorthogonal quantum states cannot be perfectly discriminated \cite{helstrom1969quantum,Bennett2014,Gisin2002,RENNER2008,Scarani2009,Ghoreishi2025a,PARIS2009,Liu2019,Montenegro2025}. 
Consequently, the problem of quantum state discrimination has been extensively studied. Based on the figure of merit, one can consider different approaches, including the framework of minimum–error state discrimination (MESD), where the goal is to maximize the probability of correctly identifying an unknown state drawn from a known ensemble, and unambiguous state discrimination (USD) where the goal is to identify the state without error, at the price of a nonzero probability of obtaining an inconclusive outcome for nonorthogonal states. Based on the problem at hand, a combination of these figures of merit can be considered. See the references \cite{Chefles2000,Barnett2009,Bae2015} for a review of different approaches.

Beyond its information–theoretic significance, state discrimination also plays an important role in quantum foundations. 
In particular, the impossibility of perfectly distinguishing nonorthogonal states has been linked to fundamental structural features of quantum theory, such as contextuality and the epistemic interpretation of quantum states \cite{Spekkens2005,Schmid2018a,Lostaglio2020,Budroni2022,Araujo2013,Borges2014}. 
Recent work has shown that the statistics of state discrimination tasks can be used to derive experimentally testable inequalities that are satisfied by all preparation–noncontextual ontological models but violated by quantum theory \cite{Schmid2018,Shin2021,Flatt2022,Mukherjee2022,RochiCarceller2024,Namkung2026,Flatt2026}.

While the standard state discrimination problem provides a natural operational framework for studying distinguishability, it is not the only way to quantify how well two states can be distinguished. In fact the number of problems that can be solved analytically is limited. As an example of analytically solved problems, we can refer to the Helstrom formula for minimum error discrimination \cite{helstrom1969quantum}, minimum-error discrimination of qubits \cite{Bae2013,Ha2013,Weir2017,Ha2023,rouhbakhsh2023geometric,Lue2026}, and some special cases such as \cite{Ban1997,Andersson2002,Chou2004,Ghoreishi2021,Lue2023}. For USD, beyond the case of two pure states \cite{Ivanovic1987,Dieks1988,Peres1988,Jaeger1995}, one can refer to the references in \cite{Chefles1998,Jafarizadeh2008,Sugimoto2010,Bergou2012,Ha2015}. 
Due to the difficulty of solving a general state discrimination problem, in many contexts, it is desirable to characterize distinguishability through alternative operational quantities that are independent of a specific (standard) discrimination strategy. 
One such quantity is the \emph{discriminability} introduced in Ref.~\cite{ghoreishi2019parametrization}, which is defined in terms of the fidelity between a state constructed from the ensemble and a classical prior state.  For two pure states, this quantity admits a closed analytic expression and coincides with the optimal discrimination success probability in the equal–prior case.  However, the operational meaning of this quantity has not yet been fully explored.

The goal of the present work is to provide an operational interpretation of discriminability and to investigate its connection to preparation contextuality. 
Specifically, we introduce an operational game that quantifies distinguishability using only experimentally accessible statistics obtained from two-copy comparison measurements. 
The score of this game defines an operational quantity $D_{\mathrm{op}}$ that can be estimated directly from measurement frequencies. 
In the qubit realization, we show that this operational score coincides with the discriminability defined in Ref.~\cite{ghoreishi2019parametrization}, thereby providing a direct experimental interpretation of this quantity.

Using this operational formulation, we then investigate the constraints imposed by preparation noncontextuality. 
We derive a direct noncontextual upper bound on the discriminability game under a SWAP-like comparison rule and a sharp single-copy test. 
We show, however, that in the natural qubit realization this bound is saturated by quantum theory. 
Thus, although the discriminability game admits a well-defined noncontextual characterization, it does not by itself yield a contextuality witness in this regime.

This observation motivates a different perspective. 
Rather than asking whether discriminability directly reveals contextuality, we investigate whether it can act as an operational resource controlling the strength of nonclassical correlations. 
To this end, we consider a Bell-coupled scenario in which two-copy comparison measurements are applied to Bob’s conditional preparations. 
We derive a state-dependent upper bound on the CHSH value in terms of operationally measurable separation parameters, and hence in terms of the distinguishability of the conditional states. 
This establishes a quantitative link between discriminability and Bell-type nonclassical behavior.

The paper is organized as follows. In Sec.~\ref{sec:operational}, we introduce the operational framework and review the notion of preparation noncontextuality in prepare--measure scenarios. In Sec.~\ref{sec:Dop_game} we define the operational discriminability game and show how it reproduces the fidelity-based discriminability in the qubit realization. In Sec.~\ref{sec:direct_bound}, we analyze the constraints imposed by preparation noncontextuality and derive a direct upper bound on the game, showing that it is saturated in the natural quantum realization. In Sec.~\ref{sec:CHSH_route} we move to a Bell-coupled scenario and derive a state-dependent CHSH bound in terms of operational separation parameters, establishing a connection between discriminability and nonclassical correlations. Finally, Sec.~\ref{sec:conclusion} concludes with a summary and outlook.

\section{Operational framework and preparation noncontextuality}
\label{sec:operational}

\subsection{Prepare--measure experiments}
Consider a finite prepare--and-measure scenario, an operational theory then specifies a set of
preparation procedures $\mathcal{P}$, a set of measurement settings $\mathcal{M}$, and conditional
probabilities 
\begin{equation}
p(k|M,P),\qquad P\in\mathcal{P},\ M\in\mathcal{M},
\end{equation}
for each outcome $k$ of measurement $M$ when performed on a system prepared by $P$.

Two preparations $P,P'\in\mathcal{P}$ are said to be \emph{operationally equivalent}, denoted
\begin{equation}
P\simeq P',
\end{equation}
if they cannot be distinguished by any measurement in the scenario, i.e.,
\begin{equation}
p(k|M,P)=p(k|M,P')\quad \forall\,M\in\mathcal{M},\ \forall\,k.
\label{eq:op_equiv_def}
\end{equation}
Operational equivalence relations among preparations are the key ingredients by which noncontextuality becomes
testable.

\subsection{Ontological models}
An \emph{ontological model} supplements the operational description with an ontic state space $\Lambda$, where
each $\lambda\in\Lambda$ represents a complete specification of the physical state of the system.
Each preparation $P$ is represented by a normalized probability distribution $\mu_P(\lambda)$ over $\Lambda$,
and each measurement outcome $k$ of measurement $M$ is represented by a response function
$\xi_{k|M}(\lambda)\in[0,1]$ satisfying $\sum_k\xi_{k|M}(\lambda)=1$ for all $\lambda$.
The operational probabilities are reproduced as
\begin{equation}
p(k|M,P)=\int_{\Lambda} d\lambda\ \xi_{k|M}(\lambda)\,\mu_P(\lambda).
\label{eq:ont_rep}
\end{equation}
Preparation noncontextuality is then the assumption that operationally equivalent preparations are represented by
the same epistemic state:
\begin{equation}
P\simeq P' \ \Rightarrow\ \mu_P(\lambda)=\mu_{P'}(\lambda)\quad \forall\,\lambda\in\Lambda.
\label{eq:pnc_def}
\end{equation}
In particular, if two mixtures of preparations are operationally equivalent, then preparation noncontextuality
imposes a nontrivial constraint on the corresponding mixtures of epistemic states.
For instance, given preparations $P$ and $P'$, define the 50--50 mixture
\begin{equation}
P_{(P+P')/2} := \tfrac12 P + \tfrac12 P',
\end{equation}
which means that in each run, choose $P$ or $P'$ uniformly at random and discard the label.
Operationally, this implies
\begin{equation}
p(k|M,P_{(P+P')/2}) = \tfrac12 p(k|M,P) + \tfrac12 p(k|M,P')\quad \forall\,M,k.
\end{equation}
In any ontological model, mixtures are represented by the corresponding convex mixtures of epistemic states:
\begin{equation}
\mu_{(P+P')/2}(\lambda)=\tfrac12 \mu_P(\lambda)+\tfrac12\mu_{P'}(\lambda).
\label{eq:mixture_mu}
\end{equation}

The primary operational equivalence used throughout this work is the existence of two distinct decompositions of
the same mixed preparation.  We assume that there exist four preparation procedures
\begin{equation}
P_T,\quad P_{\tilde T},\quad P_\eta,\quad P_{\tilde\eta}
\end{equation}
such that their 50--50 mixtures are operationally equivalent:
\begin{equation}
\tfrac12 P_T+\tfrac12 P_{\tilde T}\ \simeq\ \tfrac12 P_\eta+\tfrac12 P_{\tilde\eta}.
\label{eq:key_op_equiv}
\end{equation}
In quantum realizations on a qubit, this equivalence corresponds to two different decompositions of the maximally
mixed state $I/2$; however, Eq.~\eqref{eq:key_op_equiv} is stated purely operationally and is the only structural
assumption needed for our noncontextuality analysis.

By preparation noncontextuality \eqref{eq:pnc_def} and mixture representation \eqref{eq:mixture_mu}, the operational
equivalence \eqref{eq:key_op_equiv} implies the ontic constraint
\begin{equation}
\tfrac12 \mu_T(\lambda)+\tfrac12\mu_{\tilde T}(\lambda)
=
\tfrac12 \mu_\eta(\lambda)+\tfrac12\mu_{\tilde\eta}(\lambda)
\quad \forall\,\lambda\in\Lambda.
\label{eq:key_mu_constraint}
\end{equation}
This is the preparation-noncontextuality ``lever'' that will yield nontrivial inequalities.

\section{Operational discriminability game}
\label{sec:Dop_game}
\subsection{Fidelity-based discriminability}
\label{sec:D_GAS}

We begin by recalling the fidelity-based notion of discriminability introduced in~\cite{ghoreishi2019parametrization}. Consider an ensemble of two pure states
$\{(\eta_1,|\psi\rangle),(\eta_2,|\phi\rangle)\}$ with priors $\eta_1,\eta_2$ (where $\eta_1+\eta_2=1$) and
overlap
\begin{equation}
\gamma_{12}:=\langle\psi|\phi\rangle.
\end{equation}
Choose an orthonormal basis $\{|e_1\rangle,|e_2\rangle\}$ such that $|e_1\rangle:=|\psi\rangle$ and
\begin{equation}
|\phi\rangle = \gamma_{12}\,|e_1\rangle + \sqrt{1-|\gamma_{12}|^2}\,|e_2\rangle,
\label{eq:phi_decomp}
\end{equation}

In \cite{ghoreishi2019parametrization}, they associate to this ensemble a $2\times 2$ density operator $\rho_T$ defined in this basis by
\begin{equation}
\rho_T
=
\begin{pmatrix}
\eta_1+\eta_2|\gamma_{12}|^2
&
\eta_2\,\gamma_{12}\sqrt{1-|\gamma_{12}|^2}
\\[4pt]
\eta_2\,\gamma_{12}^*\sqrt{1-|\gamma_{12}|^2}
&
\eta_2\big(1-|\gamma_{12}|^2\big)
\end{pmatrix}.
\label{eq:rhoT_explicit}
\end{equation}
Equivalently, $\rho_T$ is the Gram-type state built from the priors and inner products of the ensemble.
They also define the diagonal \emph{prior state}
\begin{equation}
\eta:=\mathrm{diag}(\eta_1,\eta_2),
\label{eq:eta_diag}
\end{equation}
and allow for relabelings of the priors via permutations $p$ of the diagonal entries, denoting the resulting
state by $\eta_p$ (for two outcomes, $p$ is either the identity or the swap).

The discriminability is then defined as the Uhlmann fidelity between $\rho_T$ and the best relabeling of the
prior state:
\begin{equation}
D(\eta) := \max_{p}\,F(\rho_T,\eta_p),
\label{eq:D_GAS_def}
\end{equation}
where the fidelity is
\begin{equation}
F(\rho,\sigma):=\left(\mathrm{Tr}\sqrt{\sqrt{\rho}\,\sigma\,\sqrt{\rho}}\right)^2.
\end{equation}

For two pure states, a closed-form expression for $D(\eta)$ in terms of the priors and the overlap
parameter can be derived
\begin{align}
D(\eta_1,\eta_2;|\gamma_{12}|^2)
&=
\eta_1^2+\eta_2^2
+2\eta_1\eta_2\sqrt{1-|\gamma_{12}|^2} \nonumber \\
&+|\gamma_{12}|^2\big(\eta_1\eta_2-\eta_{\min}^2\big),
\label{eq:D_closed_form}
\end{align}
where $\eta_{\min}:=\min\{\eta_1,\eta_2\}$.  In the special case of equal priors $\eta_1=\eta_2=\tfrac12$,
this discriminability coincides with the optimal success probability for minimum-error discrimination of the
two states; for unequal priors it defines a distinct, fidelity-based measure of distinguishability.

In the remainder of this section, we introduce an operational game, defined entirely in terms of observed
probabilities in a fixed laboratory experiment, whose score reproduces $D(\eta)$ in the qubit realization.
\subsection{Operational discriminability score}
We introduce a fixed binary-outcome measurement on two systems, denoted $M_{\mathrm{swap}}$, with outcomes
labeled \textsf{pass} and \textsf{fail}.Operationally, $M_{\mathrm{swap}}$ is simply a laboratory procedure that takes two prepared systems and outputs a bit.
(In a qubit quantum realization, this can be implemented as a SWAP test.) We define two primitive experiments, each producing an empirical \textsf{pass} frequency. In the first experiment, denoted \(G_{\mathrm{mix}}^{(p)}\), each run proceeds as follows:
\begin{enumerate}
\item Prepare system $A$ using $P_T$.
\item Prepare system $B$ using $P_\eta^{(p)}$, where $p$ is a fixed labeling (permutation) of the two prior outcomes.
\item Apply $M_{\mathrm{swap}}$ to the pair $(A,B)$ and record whether the outcome is \textsf{pass}.
\end{enumerate}
Let
\begin{equation}
p_{\mathrm{mix}}^{(p)} := p(\textsf{pass}\mid P_T,P_\eta^{(p)};M_{\mathrm{swap}})
\label{eq:pmix}
\end{equation}
denote the observed \textsf{pass} probability.

In the second experiment, denoted \(G_{\mathrm{pur}}\), each run proceeds as follows:
\begin{enumerate}
\item Prepare system $A$ using $P_T$.
\item Prepare system $B$ independently using $P_T$.
\item Apply $M_{\mathrm{swap}}$ to the pair $(A,B)$ and record \textsf{pass}.
\end{enumerate}
Let
\begin{equation}
p_{\mathrm{pur}} := p(\textsf{pass}\mid P_T,P_T;M_{\mathrm{swap}})
\label{eq:ppur}
\end{equation}
denote the observed \textsf{pass} probability.
Now assume that the prior probabilities $(\eta_1,\eta_2)$ are known a priori.
(For the two-outcome case, there are two possible labelings $p\in S_2$.)
We define the \emph{operational discriminability} as a scored game whose score is a function of the
experimentally observed \textsf{pass} rates \eqref{eq:pmix}--\eqref{eq:ppur}:

\begin{definition}[Operational discriminability $D_{\mathrm{op}}(\eta)$]
\label{def:Dop}
Define
\begin{equation}
D_{\mathrm{op}}(\eta)
:= \max_{p\in S_2}
\Big[
(2p_{\mathrm{mix}}^{(p)}-1)
\;+\;
2\sqrt{\eta_1\eta_2(1-p_{\mathrm{pur}})}
\Big].
\label{eq:Dop}
\end{equation}
\end{definition}

Since this definition refers only to preparation procedures, and a fixed measurement, and observed probabilities, it is purely operational.

Let us stress that no state discrimination (or any guessing task) is invoked in the definition. However, in a qubit quantum realization where $M_{\mathrm{swap}}$ is the swap test and $P_\eta^{(p)}$
prepares the diagonal prior state, one can show that
\begin{equation}
2p_{\mathrm{mix}}^{(p)}-1=\mathrm{Tr}(\rho_T\eta_p),\qquad
2p_{\mathrm{pur}}-1=\mathrm{Tr}(\rho_T^2),
\end{equation}
and for qubits
\[
F(\rho,\sigma)=\mathrm{Tr}(\rho\sigma)+2\sqrt{\det\rho\,\det\sigma},
\]
leading to $D_{\mathrm{op}}(\eta)=\max_p F(\rho_T,\eta_p)$.
This interpretation is not required for operational validity but is useful for connecting to standard quantities.
%

%

\section{Direct preparation-noncontextual bound for the discriminability game}
\label{sec:direct_bound}

In Secs.~\ref{sec:operational} and \ref{sec:Dop_game}, we introduced the operational discriminability game
and the preparation-noncontextuality framework in which to analyze it.
A natural question is whether the score $D_{\mathrm{op}}(\eta)$ itself can serve as a direct witness of
preparation contextuality.  In other words, does preparation noncontextuality imply a nontrivial upper bound on
the achievable value of the discriminability game?

In this section we address this question in the simplest nontrivial ontological model.
We begin by formulating the game directly at the ontological level.
We then derive a direct preparation-noncontextual bound under a natural SWAP-like restriction on the comparison
measurement and an ideal sharpness assumption for a single-copy test.
Finally, we show that the resulting bound is saturated by the natural qubit realization, so that the direct route
does not by itself yield a contextuality witness in this regime.
This negative result motivates the more powerful CHSH-coupled construction developed in the next section.

\subsection{Ontological formulation of the two-copy discriminability game}
\label{sec:direct_ontological}

We consider a finite ontological model with single-system ontic state space
\begin{equation}
\Lambda=\{1,\dots,n\}.
\end{equation}
Each preparation $P\in\{P_T,P_{\tilde T},P_\eta,P_{\tilde\eta}\}$ is represented by an epistemic state
\begin{equation}
\mu_P(i)\ge 0,\qquad \sum_{i=1}^{n}\mu_P(i)=1.
\label{eq:mu_norm_direct}
\end{equation}

The key preparation-noncontextuality assumption is the operational equivalence of the 50--50 mixtures
\begin{equation}
\tfrac12 P_T+\tfrac12 P_{\tilde T}\simeq \tfrac12 P_\eta+\tfrac12 P_{\tilde\eta},
\end{equation}
which, by preparation noncontextuality, implies
\begin{equation}
\mu_T(i)+\mu_{\tilde T}(i)=\mu_\eta(i)+\mu_{\tilde\eta}(i)
\qquad \forall i\in\Lambda.
\label{eq:pnc_direct_constraint}
\end{equation}

The comparison game uses a binary two-copy measurement
\[
M_{\mathrm{swap}}=\{\textsf{pass},\textsf{fail}\}
\]
applied to two independently prepared systems.
In an ontological model, this measurement is represented by a response function on pairs of ontic states,
\begin{equation}
\xi_{\mathrm{pass}}(i,j)\in[0,1].
\label{eq:xi_pair_general}
\end{equation}
If the comparison measurement is exchange-symmetric, then one may impose
\[
\xi_{\mathrm{pass}}(i,j)=\xi_{\mathrm{pass}}(j,i).
\]

The primitive probabilities appearing in the discriminability game are then
\begin{align}
p_{\mathrm{mix}}^{(p)}
&=
\sum_{i,j}\xi_{\mathrm{pass}}(i,j)\,\mu_T(i)\,\mu_{\eta}^{(p)}(j),
\label{eq:pmix_general_direct}
\\
p_{\mathrm{pur}}
&=
\sum_{i,j}\xi_{\mathrm{pass}}(i,j)\,\mu_T(i)\,\mu_T(j),
\label{eq:ppur_general_direct}
\end{align}
where $\mu_\eta^{(p)}$ denotes the relabeling the prior preparation used in the game.
The operational discriminability score is
\begin{equation}
D_{\mathrm{op}}(\eta)
=
\max_{p\in S_2}
\left[
(2p_{\mathrm{mix}}^{(p)}-1)
+
2\sqrt{\eta_1\eta_2{(1-p_{\mathrm{pur}})}} \right].
\label{eq:Dop_recalled_direct}
\end{equation}

To obtain a first analytic handle, we consider the smallest nontrivial ontic space,
\begin{equation}
\Lambda=\{0,1\}.
\end{equation}
Each preparation is then a Bernoulli distribution:
\begin{align}
\mu_T &= (t,1-t), &
\mu_{\tilde T} &= (\tilde t,1-\tilde t), \nonumber\\
\mu_\eta &= (e,1-e), &
\mu_{\tilde\eta} &= (\tilde e,1-\tilde e),
\end{align}
with $t,\tilde t,e,\tilde e\in[0,1].$ The preparation-noncontextuality constraint \eqref{eq:pnc_direct_constraint} is then reduced to
\begin{equation}
t+\tilde t=e+\tilde e.
\label{eq:n2_pnc}
\end{equation}

The exchange-symmetric two-copy response is a $2\times 2$ matrix
\begin{equation}
\Xi=
\begin{pmatrix}
a & b\\
b & c
\end{pmatrix},
\qquad a,b,c\in[0,1].
\label{eq:Xi_general}
\end{equation}
The corresponding primitive probabilities are
\begin{align}
p_{\mathrm{pur}}
&=
a\,t^2 + 2b\,t(1-t)+c\,(1-t)^2,
\label{eq:ppur_n2_general}
\\
p_{\mathrm{mix}}
&=
a\,te+b\,t(1-e)+b\,(1-t)e+c\,(1-t)(1-e).
\label{eq:pmix_n2_general}
\end{align}


However, the most general response matrix \eqref{eq:Xi_general} is too unconstrained to carry an immediate physical
interpretation as a similarity test. We therefore impose the natural SWAP-like restriction that identical ontic states pass with certainty, while
different ontic states pass with a common probability $q\in[0,1)$:
\begin{equation}
\Xi_q=
\begin{pmatrix}
1 & q\\
q & 1
\end{pmatrix}.
\label{eq:Xi_q}
\end{equation}
This is the simplest classical analogue of a comparison measurement that rewards ontic similarity:
equal ontic states are maximally similar, while different ontic states pass only with the reduced probability $q$.
Under \eqref{eq:Xi_q}, Eqs.~\eqref{eq:ppur_n2_general}--\eqref{eq:pmix_n2_general} simplify to
\begin{align}
p_{\mathrm{pur}}
&=
t^2+2q\,t(1-t)+(1-t)^2
=
1-2(1-q)t(1-t),
\label{eq:ppur_q}
\\
p_{\mathrm{mix}}
&=
te+q\,t(1-e)+q(1-t)e+(1-t)(1-e)
\nonumber\\
&=
1-(1-q)\big[t(1-e)+(1-t)e\big].
\label{eq:pmix_q}
\end{align}
It is convenient to define the disagreement probability
\begin{equation}
\delta:=t(1-e)+(1-t)e=t+e-2te,
\label{eq:delta_def}
\end{equation}
so that
\begin{equation}
p_{\mathrm{mix}}=1-(1-q)\delta.
\label{eq:pmix_delta}
\end{equation}

Substituting \eqref{eq:ppur_q} and \eqref{eq:pmix_delta} into \eqref{eq:Dop_recalled_direct} yields
\begin{equation}
D_{\mathrm{op}}
=
1-2(1-q)\delta
+
2\sqrt{2\eta_1\eta_2(1-q)t(1-t)}.
\label{eq:Dop_q_general}
\end{equation}
At first sight, one might hope that preparation noncontextuality together with the operational equivalence of
mixtures would already force a nontrivial upper bound on $D_{\mathrm{op}}$.
However, this is not the case, even under the natural SWAP-like restriction \eqref{eq:Xi_q}. Indeed, choosing
$t=e=1$ gives
$\delta=0,\qquad t(1-t)=0,$ and therefore
\[
p_{\mathrm{mix}}=1,\qquad p_{\mathrm{pur}}=1,
\]
which implies $D_{\mathrm{op}}=1.$
This choice is compatible with the preparation-noncontextuality condition \eqref{eq:n2_pnc}, since one can choose
the complementary preparations appropriately so that the mixture equivalence is satisfied.

Thus, even under the SWAP-like restriction, the preparation-noncontextual maximum remains trivial:
\[
D_{\mathrm{NC}}^{\mathrm{direct}}=1.
\]

This shows that preparation noncontextuality, together with the mixture equivalence alone, is insufficient to
generate a nontrivial discriminability witness.  One needs an additional operational ingredient certifying that
the preparations $P_T$ and $P_\eta$ are genuinely distinct. For this purpose,
we now introduce the minimal additional ingredient needed to obtain a nontrivial direct bound.
Assume that there exists a sharp single-copy test
\[
M_T=\{T,\tilde T\}
\]
such that
\begin{equation}
p(T|M_T,P_T)=1,\qquad p(T|M_T,P_{\tilde T})=0.
\label{eq:sharp_test}
\end{equation}
This ideal sharpness condition plays the same role as the perfect predictability assumptions in the noiseless
analysis of Schmid and Spekkens \cite{Schmid2018}.

In the $n=2$ ontological model, \eqref{eq:sharp_test} forces the supports of $\mu_T$ and $\mu_{\tilde T}$ to be
disjoint.  Since there are only two ontic states, this implies, up to relabeling,
\begin{equation}
\mu_T=(1,0),\qquad \mu_{\tilde T}=(0,1).
\label{eq:sharp_supports}
\end{equation}
The preparation-noncontextuality condition \eqref{eq:pnc_direct_constraint} then forces
\begin{equation}
\mu_\eta=(c,1-c),\qquad \mu_{\tilde\eta}=(1-c,c)
\label{eq:eta_conf}
\end{equation}
for some $c\in[0,1]$, where
\begin{equation}
c:=p(T|M_T,P_\eta)
\label{eq:c_confusability}
\end{equation}
is the confusability of $P_\eta$ relative to the sharp test for $P_T$.
Using \eqref{eq:sharp_supports} and \eqref{eq:eta_conf}, the primitive probabilities become
\begin{align}
p_{\mathrm{pur}}
&=
1,
\label{eq:ppur_sharp}
\\
p_{\mathrm{mix}}^{(\mathrm{id})}
&=
q+(1-q)c.
\label{eq:pmix_sharp}
\end{align}
Since $p_{\mathrm{pur}}=1$, the square-root term in \eqref{eq:Dop_recalled_direct} vanishes, and the score reduces
to a linear function of the mixed pass probability.
This yields the first nontrivial direct bound.

\begin{theorem}[Direct noncontextual upper bound for the discriminability game]
\label{thm:direct_bound}
Consider the $n=2$ ontological model with SWAP-like comparison response \eqref{eq:Xi_q}.
Assume further that there exists a sharp single-copy test satisfying \eqref{eq:sharp_test}, and that the four
preparations obey the preparation-noncontextual mixture equivalence \eqref{eq:key_op_equiv}.
Let
\[
c:=p(T|M_T,P_\eta)
\]
be the corresponding confusability.
Then the operational discriminability score satisfies
\begin{equation}
D_{\mathrm{op}}
\le
1-2(1-q)\min\{c,1-c\}.
\label{eq:direct_bound_full}
\end{equation}
In particular, after relabeling so that $c\le \tfrac12$,
\begin{equation}
D_{\mathrm{op}}\le 1-2(1-q)c.
\label{eq:direct_bound_relabel}
\end{equation}
\end{theorem}

\begin{proof}
By \eqref{eq:sharp_supports}, $\mu_T=(1,0)$ and $\mu_{\tilde T}=(0,1)$.
Preparation noncontextuality then implies \eqref{eq:eta_conf}.
Substituting these distributions into the SWAP-like response \eqref{eq:Xi_q} gives
\[
p_{\mathrm{pur}}=1,\qquad
p_{\mathrm{mix}}^{(\mathrm{id})}=q+(1-q)c.
\]
Therefore
\[
D_{\mathrm{op}}^{(\mathrm{id})}
=
2p_{\mathrm{mix}}^{(\mathrm{id})}-1
=
2q-1+2(1-q)c.
\]
For the swapped labeling one obtains
\[
p_{\mathrm{mix}}^{(\mathrm{swap})}=1-(1-q)c.
\]
Hence
\[
D_{\mathrm{op}}
=
2\max\{p_{\mathrm{mix}}^{(\mathrm{id})},p_{\mathrm{mix}}^{(\mathrm{swap})}\}-1
=
1-2(1-q)\min\{c,1-c\},
\]
which proves the claim.
\end{proof}

As a special case of the ideal classical similarity test, if $q=0$, so that identical ontic states always pass and different ontic states always fail, then
\begin{equation}
D_{\mathrm{op}}\le 1-2c
\qquad (c\le \tfrac12).
\label{eq:direct_bound_q0}
\end{equation}

Theorem~\ref{thm:direct_bound} gives a genuine nontrivial upper bound on the discriminability game.
However, in the natural quantum realization with a qubit sharp test and the standard SWAP measurement, this bound is tight and is saturated by quantum theory.

To see this, take
\[
\rho_T = |0\rangle\langle 0|,\qquad
\eta=
\begin{pmatrix}
c & 0\\
0 & 1-c
\end{pmatrix},
\]
and let
\[
M_T=\{|0\rangle\langle 0|,\ |1\rangle\langle 1|\}
\]
be the sharp single-copy test.
Then indeed
\[
p(T|M_T,P_\eta)=\langle 0|\eta|0\rangle = c.
\]

For the standard quantum SWAP test,
\[
p_{\mathrm{pass}}(\rho,\sigma)=\frac{1+\mathrm{Tr}(\rho\sigma)}{2}.
\]
Since $\rho_T$ is pure,
\[
p_{\mathrm{pur}}=p_{\mathrm{pass}}(\rho_T,\rho_T)=1.
\]
Moreover,
\[
p_{\mathrm{pass}}(|0\rangle\langle 0|,|1\rangle\langle 1|)=\frac12,
\]
so the natural quantum SWAP test corresponds to $q=\frac12$.

Using
\[
p_{\mathrm{mix}}^{(\mathrm{id})}=\frac{1+c}{2},
\qquad
p_{\mathrm{mix}}^{(\mathrm{swap})}=\frac{2-c}{2},
\]
one finds
\begin{equation}
D_{\mathrm{op}}^{\mathrm{QM}}=1-c,
\label{eq:QM_D_direct}
\end{equation}
which exactly equals the bound \eqref{eq:direct_bound_relabel} with $q=\tfrac12$. 

It is also useful to compare the direct noncontextual bound \eqref{eq:direct_bound_relabel} with the closed-form discriminability \eqref{eq:D_closed_form}. For any nontrivial two-state ensemble with \(0<|\gamma_{12}|^2<1\), one has \(D(\eta_1,\eta_2;|\gamma_{12}|^2)<1\). Hence, for any fixed \(c>0\), there always exists \(q<1\) such that
\[
1-2(1-q)c \ge D(\eta_1,\eta_2;|\gamma_{12}|^2).
\]
In particular, if one sets \(c=\eta_{\min}\), the threshold value is
\[
q_* = 1-\frac{1-D(\eta_1,\eta_2;|\gamma_{12}|^2)}{2\eta_{\min}}<1.
\]
Therefore, the direct bound does not lead to a contradiction with the discriminability defined in \cite{ghoreishi2019parametrization}, and no direct contextuality conclusion follows from this comparison.

It is important to emphasize that Theorem~\ref{thm:direct_bound} is derived within the minimal ontological model with \(n=2\). For larger ontic state spaces, the sharp single-copy test does not force \(P_T\) and \(P_{\tilde T}\) to be represented by point distributions, but only implies that their epistemic states have disjoint supports. Consequently, \(p_{\mathrm{pur}}\) need not equal \(1\), and the square-root term in \(D_{\mathrm{op}}\) can remain nonzero. Thus, Theorem~1 should be interpreted as a minimal-model bound rather than as a general preparation-noncontextual inequality.

Nevertheless, the qualitative message of the minimal analysis is plausibly robust: enlarging the ontic space gives a preparation-noncontextual model additional freedom in the choice of epistemic states and response functions. One should therefore not expect a direct contextuality violation to be restored simply by moving beyond \(n=2\). This supports the conclusion that the direct discriminability game is structurally limited as a standalone contextuality witness, and motivates the Bell-coupled construction of the next section.

The results of this section clarify the precise status of the discriminability game as a contextuality witness. 
On the one hand, Theorem~\ref{thm:direct_bound} shows that, once supplemented with a sharp single-copy test and a SWAP-like comparison measurement, the game admits a nontrivial preparation-noncontextual upper bound. 
Thus, the operational score $D_{\mathrm{op}}$ is not unconstrained: its achievable values are quantitatively limited by the confusability parameter $c$.

On the other hand, we show that in the natural qubit realization—namely, a sharp projective test together with the standard quantum SWAP comparison measurement—this bound is saturated by quantum theory and direct noncontextual bound does not lead to a contradiction with the discriminability defined in Eq.~\eqref{eq:D_closed_form}, and hence does not provide a contextuality witness for generic two-state ensembles.

It is important to note that this is not a limitation of a particular implementation but reflects a structural feature of the scenario. 
The noncontextual bound derived here relies on the existence of an ideal sharp single-copy test, which effectively trivializes the two-copy purity term and reduces the game to a classical similarity test. 
As a result, the natural quantum realization achieves the maximal value allowed by preparation noncontextuality.

This observation indicates that discriminability, while operationally well-defined and nontrivially constrained, is not by itself sufficient to reveal contextuality in the present prepare--measure setting. 
To obtain a genuine quantum advantage, one must couple discriminability to additional operational resources. 
In the next section, we show that Bell-type correlations provide precisely such a resource, leading to a nontrivial connection between distinguishability and nonclassical correlation strength.

\section{Discriminability as a resource for CHSH correlations}
\label{sec:CHSH_route}

In Sec.~\ref{sec:direct_bound}, we showed that the discriminability game admits a direct
preparation-noncontextual characterization once one supplements the scenario with a sharp single-copy test and
a SWAP-like comparison measurement.  However, it was also shown that in the natural
qubit realization this direct bound is saturated by quantum theory.  Thus, although the game is operationally
well-defined and nontrivially constrained, it does not by itself witness contextuality in that regime.

This naturally suggests a different question: rather than asking whether the discriminability game alone is
contextual, can discriminability instead act as a resource parameter controlling the strength of
nonclassical correlations?  In this section, we answer this question affirmatively by coupling the SWAP-game
statistics to a Bell-type correlation experiment.  The resulting inequality provides a state-dependent upper bound
on the achievable CHSH value in terms of operationally estimated separation parameters, which can in turn be
interpreted as discriminability measures for Bob's conditional preparations.

\subsection{Bipartite correlation experiment and conditional preparations}
\label{sec:bipartite_conditional}

Consider a bipartite experiment with two parties, Alice and Bob.
Alice chooses a setting $x\in\{0,1\}$ and obtains an outcome $a\in\{\pm1\}$, while Bob chooses a setting
$y\in\{0,1\}$ and obtains an outcome $b\in\{\pm1\}$.
The observed correlations are described by the conditional probabilities $p(a,b|x,y)$
from which one defines the correlators
\begin{equation}
E_{xy}:=\sum_{a,b=\pm1}ab\,p(a,b|x,y).
\label{eq:Exy_def}
\end{equation}
The corresponding CHSH expression is \cite{CHSH}
\begin{equation}
S_{\mathrm{CHSH}}:=E_{00}+E_{01}+E_{10}-E_{11}.
\label{eq:CHSH_def_rewrite}
\end{equation}

Operationally, conditioning on Alice's setting $x$ and outcome $a$ defines a subensemble on Bob.  We therefore treat Alice's event \((a,x)\) as specifying a conditional preparation on Bob, denoted \(P_{a|x}\), i.e., Bob's preparation conditioned on Alice's setting \(x\) and outcome \(a\).

This steering-style viewpoint is purely operational. It does not require any quantum formalism, only the
ability to postselect Bob's runs according to Alice's recorded outcome.
Therefore, for each fixed $x$, the pair
$\{P_{+|x},\,P_{-|x}\}$
constitutes a binary state-discrimination problem on Bob.  The main goal of this section is to show that the
distinguishability of these two conditional preparations constrains the attainable Bell correlations.
For this purpose, we first introduce operational quantities that characterize the distinguishability of Bob's conditional preparations using only experimentally accessible SWAP-type statistics. For each setting $x\in\{0,1\}$, consider the two conditional preparations $P_{+|x}$ and $P_{-|x}$, associated with states $\rho_{+|x}$ and $\rho_{-|x}$. From the SWAP comparison experiment, we define the purity and overlap probabilities
\begin{align}
p_{\mathrm{pur}}^{+|x} &= p(\textsf{pass}|\rho_{+|x},\rho_{+|x}),\\
p_{\mathrm{pur}}^{-|x} &= p(\textsf{pass}|\rho_{-|x},\rho_{-|x}),\\
p_{\mathrm{ov}}^{x} &= p(\textsf{pass}|\rho_{+|x},\rho_{-|x}).
\end{align}
The first and second quantities are  purity proxies for the conditional preparation's $P_{a|x}$, while the third is an overlap
proxy between the two conditional preparations at fixed $x$.
These quantities can be combined to define an operational measure of separation between the conditional preparations.
In the general case, where the preparation probabilities
$\pi_{a|x}:=p(a|x)$ may be biased; the relevant quantity is the weighted separation parameter
\begin{align}
\widetilde R_x^2
&:=
2\Big[
\pi_{+|x}^2(2p_{\mathrm{pur}}^{+|x}-1)
+
\pi_{-|x}^2(2p_{\mathrm{pur}}^{-|x}-1) \nonumber \\
&-
2\pi_{+|x}\pi_{-|x}(2p_{\mathrm{ov}}^{x}-1)
\Big]
-
(\pi_{+|x}-\pi_{-|x})^2.
\label{eq:weighted_Rx_main}
\end{align}

This quantity depends only on experimentally accessible probabilities and provides an operational measure of how well the two conditional preparations can be distinguished via comparison measurements.

In the symmetric case $p(+|x)=p(-|x)=\tfrac12$, Eq.~\eqref{eq:weighted_Rx_main} reduces to the simpler expression
\begin{equation}
R_x^2
=
p_{\mathrm{pur}}^{+|x}
+
p_{\mathrm{pur}}^{-|x}
-
2p_{\mathrm{ov}}^{x}
 ,
\label{eq:Rx_symmetric_main}
\end{equation}

\subsection{CHSH bound from operational discriminability}
\label{sec:CHSH_bound}

We now relate the operational separation parameters to the strength of nonclassical correlations.

\begin{theorem}
\label{thm:weighted_swap_CHSH_main}
Assume a qubit realization in which Bob's conditional preparations are represented by qubit states and Bob's binary measurements are projective. Let $\widetilde R_x$ be the weighted operational separation parameter defined in Eq.~\eqref{eq:weighted_Rx_main}. Then the CHSH value satisfies
\begin{equation}
S_{\mathrm{CHSH}}
\le
2\sqrt{\widetilde R_0^2+\widetilde R_1^2}.
\label{eq:weighted_CHSH_main}
\end{equation}
\end{theorem}

A complete proof is provided in Appendix~\ref{app:weighted_swap_CHSH}. 
We emphasize that the content of Theorem~\ref{thm:weighted_swap_CHSH_main} is not the
Cauchy--Schwarz optimization alone, but the operational
identification of the relevant steering-vector norms with
two-copy comparison statistics. Thus, the bound can be
expressed in terms of experimentally accessible SWAP-type
probabilities rather than requiring full tomography of Bob's
conditional states.

To relate this bound directly to the discriminability quantity defined in Eq. \eqref{eq:D_closed_form}, we use the following lemmas.

\begin{lemma}[Geometric meaning of $\widetilde R_x$ for qubits]
\label{lem:Rtilde_geometry}
Assume Bob's conditional preparations admit a qubit realization
\[
\rho_{a|x}=\tfrac12\big(I+\bm{s}_{a|x}\cdot\bm{\sigma}\big),
\]
where $\bm{s}_{a|x}\in\mathbb{R}^3$ are Bloch vectors, and let
\[
\pi_{a|x}:=p(a|x), \qquad a\in\{\pm1\}.
\]
Define the weighted steering vector
\[
\bm{r}_x:=\pi_{+|x}\bm{s}_{+|x}-\pi_{-|x}\bm{s}_{-|x}.
\]
Then the weighted operational separation parameter satisfies
\begin{equation}
\widetilde R_x=\|\bm{r}_x\|.
\label{eq:Rtilde_Bloch}
\end{equation}
\end{lemma}

This lemma identifies the operational quantity $\widetilde R_x$ with the norm of the weighted steering vector, thereby giving it a direct geometric meaning in the Bloch representation. In the symmetric case $\pi_{+|x}=\pi_{-|x}=\tfrac12$, this further coincides with the trace norm of the weighted state difference. The proof of this lemma is also included in the appendix \ref{app:weighted_swap_CHSH} (See step 4).

\begin{lemma}
\label{lem:R_vs_D}
Suppose that, for a fixed setting \(x\), the conditional preparations
\(P_{+|x}\) and \(P_{-|x}\) are pure qubit states, so that the fidelity-based
discriminability \(D_x\) of Eq.~\eqref{eq:D_closed_form} is directly applicable to the conditional
ensemble
\[
\{(\pi_{+|x},P_{+|x}),(\pi_{-|x},P_{-|x})\}.
\]
Then the weighted operational separation parameter satisfies
\begin{equation}
\widetilde R_x \leq 2D_x-1 .
\label{eq:Rtilde_vs_D}
\end{equation}
\end{lemma}

\begin{proof}
For pure conditional qubit states, the discriminability \(D_x\) defined in Eq.~\ref{eq:D_closed_form}
is applicable to the two-state conditional ensemble. Using the result of
Ref.~\cite{ghoreishi2019parametrization} that the corresponding Helstrom guessing probability is upper bounded
by this discriminability, one has
\[
p_{\mathrm{guess},x}\leq D_x .
\]
On the other hand, for the conditional pair at setting $x$, the unequal-prior Helstrom formula gives
\[
p_{\mathrm{guess},x}
=
\frac12\left(1+\big\|\pi_{+|x}\rho_{+|x}-\pi_{-|x}\rho_{-|x}\big\|_1\right).
\]
For qubits, the operator
\[
\pi_{+|x}\rho_{+|x}-\pi_{-|x}\rho_{-|x}
=
\frac12\Big((\pi_{+|x}-\pi_{-|x})I+\bm r_x\cdot\bm\sigma\Big)
\]
has trace norm at least $\|\bm r_x\|$. Using Lemma~\ref{lem:Rtilde_geometry}, $\|\bm r_x\|=\widetilde R_x$, and therefore
\[
2p_{\mathrm{guess},x}-1\ge \widetilde R_x.
\]
Combining this with $p_{\mathrm{guess},x}\le D_x$ yields
\[
\widetilde R_x\le 2D_x-1,
\]
which proves the claim.
\end{proof}
By substituting the bound $\widetilde R_x \le 2D_x - 1$ from Lemma~\ref{lem:R_vs_D} into Eq.~\eqref{eq:weighted_CHSH_main}, the following bound is obtained immediately.
\begin{corollary}[CHSH bound in terms of discriminability]
\label{cor:CHSH_discriminability}
Under the assumptions of Theorem~\ref{thm:weighted_swap_CHSH_main}, and for the case where conditional preparations $(P_{+|x},P_{-|x})$ are represented by pure qubit states, the CHSH value satisfies
\begin{equation}
S_{\mathrm{CHSH}}
\le
2\sqrt{(2D_0-1)^2+(2D_1-1)^2}.
\label{eq:CHSH_D_bound}
\end{equation}
\end{corollary}
The discriminability form of the bound should therefore be understood as a specialization of the more general operational bound in Eq.~\eqref{eq:weighted_CHSH_main}. The latter
depends only on the directly measurable separation parameters
\(\widetilde R_x\) and remains valid for arbitrary qubit conditional
preparations. By contrast, the replacement of \(\widetilde R_x\) by
\(2D_x-1\) relies on applying the fidelity-based discriminability of Eq.~\eqref{eq:D_closed_form}
to the conditional two-state ensembles, and is therefore stated here for pure conditional preparations.

Theorem~\ref{thm:weighted_swap_CHSH_main} shows that the strength of CHSH correlations is constrained by
the operational separation of Bob's conditional preparations. In the pure-state
case, Corollary~\ref{cor:CHSH_discriminability} expresses this constraint directly in terms of the
fidelity-based discriminability \(D_x\).
Larger discriminability implies larger separation parameters $\widetilde R_x$, and hence allows for stronger nonclassical correlations.

Equation~\eqref{eq:weighted_CHSH_main} is nontrivial. In particular, if $\widetilde R_0=\widetilde R_1=1$, one recovers the Tsirelson bound $S_{\mathrm{CHSH}}\le 2\sqrt{2}$. More generally, if
$\widetilde R_0^2+\widetilde R_1^2<1,$
then $S_{\mathrm{CHSH}}<2$, so no CHSH violation is possible for any choice of Bob's measurements. This shows that insufficient discriminability of the conditional preparations precludes the appearance of nonclassical correlations.

In the symmetric case, where the two settings yield equal discriminability, $D_0=D_1=:D$, Eq.~\eqref{eq:CHSH_D_bound} reduces to
\begin{equation}
S_{\mathrm{CHSH}}\le 2\sqrt{2}\,(2D-1).
\label{eq:CHSH_D_sym}
\end{equation}
In particular, a necessary condition for CHSH violation is
\begin{equation}
D>\tfrac12\left(1+\tfrac{1}{\sqrt2}\right)\approx 0.8536.
\label{eq:D_threshold_rewrite}
\end{equation}
The results of this section provide a direct bridge between Bell nonlocality and two-state discriminability.
For each setting $x$, the quantity $D_x$ measures how well Bob could, in principle, distinguish the two
conditional preparations associated with Alice's outcomes.
Equation~\eqref{eq:CHSH_D_bound} states that the strength of Bell correlations is limited by the
distinguishability of these conditional states.
Equivalently, if Bob's conditional preparations are insufficiently distinguishable, then strong CHSH violations
are impossible.




\subsection{Conceptual significance}
\label{sec:CHSH_significance}

The results of this section suggest the following conceptual picture.

A Bell experiment naturally induces, for each Alice setting $x$, a binary discrimination problem on Bob:
Alice's two outcomes define two conditional preparations $P_{+|x}$ and $P_{-|x}$.
The distinguishability of these preparations quantifies how much information about Alice's outcome is available on
Bob's side.
Our SWAP-game construction provides a direct operational method to estimate this distinguishability using only
two-copy comparison measurements.
Theorem~\ref{thm:weighted_swap_CHSH_main} and Corollary~\ref{cor:CHSH_discriminability} then show that this operationally accessible
distinguishability limits the strength of Bell correlations.

Thus, while the direct discriminability game of Sec.~\ref{sec:direct_bound} does not itself witness contextuality
in the natural qubit realization, the Bell-coupled scenario reveals a different and stronger role for
discriminability by becoming a quantitative resource controlling nonclassical correlation strength.
This is the sense in which discriminability acquires a contextual (or Bell-nonclassical) significance in the
present framework.

\section{Conclusion}
\label{sec:conclusion}

We have studied discriminability from an operational and contextuality-oriented viewpoint. 
Starting from the fidelity-based discriminability introduced in Ref.~\cite{ghoreishi2019parametrization}, we showed that in the qubit 
realization it admits a natural operational interpretation through a two-copy comparison game based on 
SWAP-type measurements. This leads to an experimentally accessible score $D_{\mathrm{op}}$ that quantifies 
state distinguishability without relying on a minimum-error discrimination task.

We first analyzed whether this operational discriminability game can directly witness preparation contextuality. 
Within a preparation-noncontextual ontological model, we derived a direct upper bound on the game score under a 
SWAP-like comparison rule and a sharp single-copy test, and showed that this bound is saturated by the natural 
qubit realization. Thus, the direct discriminability game provides a meaningful structural characterization but 
does not by itself yield a contextuality witness in this regime.

We then showed that discriminability acquires a stronger nonclassical significance when coupled to Bell-type 
correlations. By applying two-copy comparison measurements to Bob's conditional preparations, we obtained a 
state-dependent upper bound on the CHSH value in terms of operational separation parameters, and hence in terms 
of the distinguishability of the conditional states. This establishes a quantitative link between operational 
discriminability and the strength of nonclassical correlations.

These results suggest that discriminability, although not itself a direct contextuality witness in the natural 
prepare--measure realization, functions as an operational resource that constrains Bell-type correlation 
behavior. 

\section*{Acknowledgements}
We acknowledge the Štefan Schwarz Support Fund, the project VEGA 2/0164/25 (QUAS), and the QENTAPP 09103-03-V04-00777 project.

\appendix

\section{Proof of Theorem~\ref{thm:weighted_swap_CHSH_main}}
\label{app:weighted_swap_CHSH}

In this appendix we derive the weighted CHSH bound stated in Theorem~\ref{thm:weighted_swap_CHSH_main}.

\subsection{Conditional states and Bob's observables}

Assume that Bob's conditional preparations admit a qubit realization
\begin{equation}
\rho_{a|x}
=
\frac12\left(I+\bm{s}_{a|x}\cdot\bm{\sigma}\right),
\qquad
a\in\{\pm1\},\ x\in\{0,1\},
\label{eq:rho_ax_weighted}
\end{equation}
where $\bm{s}_{a|x}\in\mathbb{R}^3$ are Bloch vectors with $\|\bm{s}_{a|x}\|\le 1$.
Let
\[
\pi_{a|x}:=p(a|x)
\]
denote Alice's conditional-outcome probabilities.

Bob's binary projective measurement for setting $y\in\{0,1\}$ is written as
\begin{equation}
B_y=\bm{b}_y\cdot\bm{\sigma},
\qquad \|\bm{b}_y\|=1.
\label{eq:By_weighted}
\end{equation}

For a qubit state $\rho=\tfrac12(I+\bm{s}\cdot\bm{\sigma})$, one has
\begin{equation}
\mathrm{Tr}(\rho B_y)=\bm{s}\cdot\bm{b}_y.
\label{eq:exp_B_weighted}
\end{equation}

\subsection{Correlators in terms of weighted steering vectors}

The correlators are
\begin{equation}
E_{xy}=\sum_{a,b=\pm1}ab\,p(a,b|x,y).
\label{eq:Exy_weighted}
\end{equation}
Using $p(a,b|x,y)=p(a|x)\,p(b|a,x,y)$ gives
\begin{align}
E_{xy}
&=
\sum_{a=\pm1} a\,p(a|x)\sum_{b=\pm1} b\,p(b|a,x,y)
\nonumber\\
&=
\sum_{a=\pm1} a\,\pi_{a|x}\,\mathrm{Tr}(\rho_{a|x}B_y).
\label{eq:Exy_step_weighted}
\end{align}
By Eq.~\eqref{eq:exp_B_weighted},
\begin{align}
E_{xy}
&=
\pi_{+|x}\,\bm{s}_{+|x}\cdot\bm{b}_y
-
\pi_{-|x}\,\bm{s}_{-|x}\cdot\bm{b}_y
\nonumber\\
&=
\bm{r}_x\cdot\bm{b}_y,
\label{eq:Exy_rx_weighted}
\end{align}
where we define the weighted steering vector
\begin{equation}
\bm{r}_x:=\pi_{+|x}\bm{s}_{+|x}-\pi_{-|x}\bm{s}_{-|x}.
\label{eq:rx_weighted}
\end{equation}
\subsection{Bounding the CHSH value}
The CHSH value is
\begin{equation}
S_{\mathrm{CHSH}}=E_{00}+E_{01}+E_{10}-E_{11}.
\label{eq:S_weighted}
\end{equation}
Using Eq.~\eqref{eq:Exy_rx_weighted},
\begin{align}
S_{\mathrm{CHSH}}
&=
\bm{r}_0\cdot\bm{b}_0+\bm{r}_0\cdot\bm{b}_1+\bm{r}_1\cdot\bm{b}_0-\bm{r}_1\cdot\bm{b}_1
\nonumber\\
&=
\bm{r}_0\cdot(\bm{b}_0+\bm{b}_1)+\bm{r}_1\cdot(\bm{b}_0-\bm{b}_1).
\label{eq:S_expand_weighted}
\end{align}
Define
\[
\bm{u}:=\bm{b}_0+\bm{b}_1,
\qquad
\bm{v}:=\bm{b}_0-\bm{b}_1.
\]
Then
\[
S_{\mathrm{CHSH}}=\bm{r}_0\cdot\bm{u}+\bm{r}_1\cdot\bm{v}.
\]
Applying the Cauchy--Schwarz inequality gives the following.
\begin{equation}
|\bm{r}_0\cdot\bm{u}+\bm{r}_1\cdot\bm{v}|
\le
\sqrt{\|\bm{r}_0\|^2+\|\bm{r}_1\|^2}\,
\sqrt{\|\bm{u}\|^2+\|\bm{v}\|^2}.
\label{eq:CS_weighted}
\end{equation}
Since $\|\bm{b}_0\|=\|\bm{b}_1\|=1$,
\begin{align}
\|\bm{u}\|^2+\|\bm{v}\|^2
&=
\|\bm{b}_0+\bm{b}_1\|^2+\|\bm{b}_0-\bm{b}_1\|^2
\nonumber\\
&=4.
\end{align}
Therefore
\begin{equation}
S_{\mathrm{CHSH}}
\le
2\sqrt{\|\bm{r}_0\|^2+\|\bm{r}_1\|^2}.
\label{eq:S_weighted_r}
\end{equation}

Thus it remains to identify $\|\bm{r}_x\|^2$ with the weighted operational quantity
$\widetilde R_x^2$.

\subsection{Relating \texorpdfstring{$\widetilde R_x$}{Rtilde x} to the weighted steering-vector norm}

Starting from Eq.~\eqref{eq:rx_weighted},
\begin{align}
\|\bm{r}_x\|^2
&=
\pi_{+|x}^2\|\bm{s}_{+|x}\|^2
+
\pi_{-|x}^2\|\bm{s}_{-|x}\|^2
-
2\pi_{+|x}\pi_{-|x}\,\bm{s}_{+|x}\cdot\bm{s}_{-|x}.
\label{eq:rnorm_expand}
\end{align}
For qubits,
\begin{equation}
\mathrm{Tr}(\rho^2)=\frac12(1+\|\bm{s}\|^2),
\qquad
\mathrm{Tr}(\rho\rho')=\frac12(1+\bm{s}\cdot\bm{s}').
\label{eq:qubit_trace_weighted}
\end{equation}
Hence
\begin{align}
\|\bm{s}_{+|x}\|^2 &= 2\mathrm{Tr}(\rho_{+|x}^2)-1,
\\
\|\bm{s}_{-|x}\|^2 &= 2\mathrm{Tr}(\rho_{-|x}^2)-1,
\\
\bm{s}_{+|x}\cdot\bm{s}_{-|x}
&=
2\mathrm{Tr}(\rho_{+|x}\rho_{-|x})-1.
\end{align}
Substituting into Eq.~\eqref{eq:rnorm_expand},
\begin{align}
\|\bm{r}_x\|^2
&=
\pi_{+|x}^2\big(2\mathrm{Tr}(\rho_{+|x}^2)-1\big)
+
\pi_{-|x}^2\big(2\mathrm{Tr}(\rho_{-|x}^2)-1\big)
\nonumber\\
&-
2\pi_{+|x}\pi_{-|x}\big(2\mathrm{Tr}(\rho_{+|x}\rho_{-|x})-1\big)
\nonumber\\
&=
2\Big[
\pi_{+|x}^2\mathrm{Tr}(\rho_{+|x}^2)
+
\pi_{-|x}^2\mathrm{Tr}(\rho_{-|x}^2)-
2\pi_{+|x}\pi_{-|x}\mathrm{Tr}(\rho_{+|x}\rho_{-|x})
\Big]
\nonumber\\
&\qquad
-
(\pi_{+|x}^2+\pi_{-|x}^2-2\pi_{+|x}\pi_{-|x})
\nonumber\\
&=
2\Big[
\pi_{+|x}^2\mathrm{Tr}(\rho_{+|x}^2)
+
\pi_{-|x}^2\mathrm{Tr}(\rho_{-|x}^2)\nonumber \\
&-
2\pi_{+|x}\pi_{-|x}\mathrm{Tr}(\rho_{+|x}\rho_{-|x})
\Big]
\nonumber\\
&\qquad
-
(\pi_{+|x}-\pi_{-|x})^2.
\label{eq:rnorm_trace}
\end{align}
Now, use the SWAP-test identities
\begin{equation}
  2p_{\mathrm{pur}}^{+|x}-1 = \mathrm{Tr}(\rho_{+|x}^2), \quad     2p_{\mathrm{pur}}^{-|x}-1 = \mathrm{Tr}(\rho_{-|x}^2),  \nonumber
\end{equation}
\begin{equation}
   2p_{\mathrm{ov}}^{x}-1 = \mathrm{Tr}(\rho_{+|x}\rho_{-|x}).
\end{equation}
Substituting into Eq.~\eqref{eq:rnorm_trace}, we obtain
\begin{align}
\|\bm{r}_x\|^2
&=
2\Big[
\pi_{+|x}^2(2p_{\mathrm{pur}}^{+|x}-1)
+
\pi_{-|x}^2(2p_{\mathrm{pur}}^{-|x}-1)
\\
&-
2\pi_{+|x}\pi_{-|x}(2p_{\mathrm{ov}}^{x}-1)
\Big]
-
(\pi_{+|x}-\pi_{-|x})^2. \nonumber 
\end{align}
This is exactly Eq.~\eqref{eq:weighted_Rx_main}, so
\[
\|\bm{r}_x\|^2=\widetilde R_x^2.
\]

Finally, substituting this into Eq.~\eqref{eq:S_weighted_r}, we arrive at
\[
S_{\mathrm{CHSH}}
\le
2\sqrt{\widetilde R_0^2+\widetilde R_1^2},
\]
which proves Theorem~\ref{thm:weighted_swap_CHSH_main}.

\qed

\bibliography{References}
\end{document}